\documentclass[lettersize,journal]{IEEEtran}
\usepackage[T1]{fontenc}
\usepackage[utf8x]{inputenc}
\usepackage{mathrsfs}
\usepackage{bm}
\usepackage{amsmath}
\usepackage{amsthm}
\usepackage{amssymb}
\usepackage{graphicx}
 \PassOptionsToPackage{normalem}{ulem}
\usepackage{hyperref}
\usepackage{bbm}
\usepackage{lipsum}
\usepackage{cleveref}
\usepackage[ruled,vlined]{algorithm2e}
\makeatletter

\theoremstyle{PDAin}
\newtheorem{theorem}{\protect\theoremname}
  \theoremstyle{PDAin}
  
  \theoremstyle{PDAin}
  
  \theoremstyle{PDAin}
  
  \theoremstyle{remark}

\theoremstyle{assumption}

\theoremstyle{algorithm}

\usepackage{amsmath}
\usepackage{amssymb}
\usepackage{graphicx,psfrag,cite,subfig}
\usepackage[table]{xcolor}

\usepackage[acronym]{glossaries}

\setkeys{Gin}{width=1.0\columnwidth}

\makeatother

\usepackage{babel}
  \providecommand{\definitionname}{Definition}
  \providecommand{\lemmaname}{Lemma}
  \providecommand{\propositionname}{Proposition}
  \providecommand{\remarkname}{Remark}
\providecommand{\theoremname}{Theorem}
\providecommand{\conjecturename}{Conjecture}
\providecommand{\assumptionname}{Assumption}

\begin{document}
 \title{Compressive Sensing Based Adaptive Defence Against Adversarial Images
 
 \thanks{The authors are with the Electrical Engineering department, Indian Institute of Technology, Delhi. Email id: akgpt7@gmail.com, \{arpanc, ee3180534\}@ee.iitd.ac.in} 
 \thanks{{\bf AKG} has contributed towards problem formulation, solution, coding and writing this paper. {\bf AC} has contributed  towards problem formulation, solution and writing. {\bf DKY} has contributed towards coding.}
 \thanks{This work was supported by the faculty seed grant, professional development allowance and professional development fund of Arpan Chattopadhyay, and the MHRD fellowship for Akash Kumar Gupta.}
 \thanks{Codes for our numerical experiments are available in~\cite{CAD}}}
 
\author{
Akash Kumar Gupta, Arpan~Chattopadhyay, Darpan Kumar Yadav
}

\maketitle


\begin{abstract}
Herein, security of deep neural network against adversarial attack is considered. Existing compressive sensing based defence schemes assume  that adversarial perturbations are usually on high frequency components, whereas  recently it has been shown that low frequency perturbations are more effective. This paper proposes a novel Compressive sensing based Adaptive Defence  (CAD) algorithm which combats  distortion in frequency domain instead of time domain. Unlike existing literature, the proposed CAD algorithm does not use information about the type of attack such as $\bm{\ell_0}$, $\bm{\ell_2}$, $\bm{\ell_\infty}$ etc. CAD algorithm uses exponential weight algorithm for exploration and exploitation to identify the type of attack,    compressive sampling matching pursuit (CoSaMP) to recover the coefficients in spectral domain, and modified basis pursuit using a  novel constraint for $\bm{\ell_0}$, $\bm{\ell_\infty}$ norm attack. Tight performance bounds for various  recovery schemes meant for various attack types are also provided. Experimental results against five state-of-the-art white box attacks on MNIST and CIFAR-10 show that the proposed CAD algorithm  achieves excellent classification accuracy and generates good quality reconstructed image with much lower computation.   
\end{abstract}


\section{Introduction}
\IEEEPARstart{T}{he} rapid development  of Deep Neural Network (DNN) and Convolutional Neural Network (CNN)   has resulted in the widening of computer vision applications such as  object recognition,  Covid-19 diagnosis  using medical images~\cite{liang2021fast},  autonomous vehicles~\cite{geiger2012we}, face detection in security and surveillance systems~\cite{jose2019face} etc. In all these applications,  images play a vital role. Recent studies have  shown  that smartly crafted, human imperceptible, small distortion in pixel values can easily fool these CNNs and DNNs~\cite{xue2021naturalae, kurakin2016adversarial, goodfellow2014explaining}. Such adversarial  images result in incorrect classification or  detection of an object or a face, leading to accidents on roads or by drones, traffic jam,  missed  identification of a criminal, etc. While  many countermeasures have been proposed in recent years to tackle adversarial images,  they are mostly based on heuristics and do not perform well against all classes of attacks. In this connection, the recent developments on  compressive sensing~\cite{candes2006stable, foucart2013invitation} allows signal recovery at sub-Nyquist rate, which is suitable for application to  images, videos and audio signals that are sparse in Fourier and wavelet domain. This also allows us to achieve lower complexity, lower power, smaller memory  and less number of sensors, and provide theoretical performance guarantee for the image processing algorithms. These reasons motivate us to use compressive sensing to combat adversarial images.

In this paper, we propose a compressive sensing based adaptive defence (CAD) algorithm  that can defend against all $l_0$, $l_2$, $l_\infty$ adversaries as well as gradient attacks. In order to identify the attack type and choose the appropriate recovery method, we use the popular exponential weight algorithm~\cite{auer1995gambling} adapted from the multi-armed bandit literature for exploration and exploitation decision, along with compressive sensing based recovery algorithms such as compressive sampling matching pursuit (CoSaMP)~\cite{needell2009cosamp}, standard basis pursuit and modified basis pursuit with novel constraints to mitigate $l_0$, $l_\infty$ attack. Numerical results reveal  that CAD is efficient in classifying both grayscale and colored images, and that it does not suffer  from clean data accuracy and gradient masking.

\subsection{Related work}
Existing research on adversarial images is broadly focused on two categories: attack design and defence algorithm design.  

{\em Attack design:} Numerous adversarial attacks have been proposed in the literature so far. They can be categorized as white box attacks  and black box attacks. In white box attacks,   the attacker has full knowledge of trained classifier its architectures, parameters and weights. Examples of white box attack include fast gradient sign method (FGSM~\cite{goodfellow2014explaining}), projected gradient descent (PGD~\cite{madry2017towards}), Carlini Wagner L2  (CW-L2) attack~\cite{carlini2017towards}, basic iterative method (BIM~\cite{kurakin2016adversarial}), Jacobian saliency map attack (JSMA~\cite{papernot2016limitations}) etc. In black box attack, the attacker generates the adversarial perturbation without having any knowledge of the target model. Transfer-based attacks~\cite{liu2016delving}, gradient estimation attacks~\cite{chen2017zoo} and boundary attack~\cite{brendel2017decision}) are some examples of black box attack. 

The adversarial attacks can also be divided into targeted attacks and non-targeted attacks. In targeted attacks,  an attacker seeks to classify an image to a target class which is different from the original class. On the other hand, in non-targeted attack, the attacker's  goal is just to misclassify an image. Based on the nature of perturbation error, attacks are further grouped into various norm attacks, such as CW ($L_2$) attack, $L_{\infty}$ BIM attack, etc.

{\em Defence design:} Adversarial image problem can be tackled either by (i) increasing the robustness of the classifier by using either image processing techniques, or adversarial training,  or compressive sensing techniques (see \cite{bafna2018thwarting,dhaliwal2020compressive, yu2018interpreting,ji2019multi,aprilpyone2020encryption}), or by  (ii) distinguishing between clean and malicious images~\cite{yadav2020efficient},~\cite{feinman2017detecting}. 

Existing defense schemes based on compressive sensing~\cite{bafna2018thwarting},~\cite{dhaliwal2020compressive} assumes that normally images have heavy spectral strength at lower frequencies and little strength at higher frequencies, which allows the adversary to modify the high-frequency spectral components to fool the human eye. Usually, most of the adversarial attacks~\cite{carlini2017towards},~\cite{szegedy2013intriguing},~\cite{brendel2017decision} work by searching the whole available attack space and are used to converge to high frequency perturbations to fool the classifier. However, it has recently been observed that constraining attack to low-frequency perturbations and keeping small distortion bound in $l_\infty$ norm is more effective, and achieves high efficiency and transferability~\cite{sharma2019effectiveness},~\cite{guo2018low}.

The authors of~\cite{bafna2018thwarting} proposed a technique based on compressive sensing to combat $l_0$ attack; the technique recovers  low frequency components corresponding to 2D discrete cosine transform (DCT) basis. In this paper, the adversarial image vector  $\bm{y}=\bm{x}+\bm{e}$, where  the original image  $\bm{x}$ is $k$-sparse in Fourier domain and the injected noise $\bm{e}$ is $t$-sparse in time domain. 
This defense is based on the fact that usually the perturbation crafted by an  attacker is on high frequency components, and hence it is not perceptible to human eye. Hence, the proposed defense works by just recovering the few top most  DCT low frequency coefficients and reconstructing images using those  coefficients only. Authors of~\cite{dhaliwal2020compressive} extended the same framework and proposed compressive recovery defense (CRD) to counter $l_2, l_{\infty}$ attack. They proposed various algorithms for different perturbation attacks which require prior knowledge of the type of attack. However, they did not   prescribe any choice of the recovery algorithm since the type  of perturbation is not known apriori.

Another popular technique to counter malicious attacks is adversarial training based defense. Here the goal is to increase the robustness of the model by training the classifier using several adversarial examples. The authors of~\cite{madry2017towards} used projected gradient adversaries and clean images to train the network; though their proposed defense works well for datasets having grayscale images such as MNIST, it   suffers from low classification accuracy for datasets having colored images such as CIFAR-10. The authors of~\cite{yu2018interpreting} used the same method and considered the properties of loss surface under various adversarial attacks in parameter and input domain. They showed that model robustness can be increased by using decision surface geometry as a parameter. The proposed defense has a very high computational complexity. The authors of~\cite{wang2020defending} proposed collaborative multi-task training (CMT) to counter various attacks. They encoded training labels into label pairs which allowed them to detect adversarial images by determining the pairwise connections between actual output and auxiliary output. However, an enormous volume of non-targeted malicious samples is needed for determining the encoding format in ~\cite{wang2020defending}. Also, the proposed defense is only applicable for non-targeted attacks. 

Several classical image processing techniques have been used earlier to combat adversarial attacks. The authors of~\cite{ji2019multi} used  Gaussian kernels with various intensities to form multiple representations of the images in the dataset, and then fed these images to the classifier. Classification and  attack detection were achieved by taking an average of multiple confidence values given by the classifier. The authors of~\cite{aprilpyone2020encryption} used pre-processing techniques; they altered the pixel values of images in the training and testing dataset block-wise by maintaining some common key. Using these image pre-processing techniques as a defense requires a lot of computations for each image in the dataset. Also, these papers did not establish  any performance bound. 

All the above papers deal with classification based defense. Detection based defense has been proposed in~\cite{yadav2020efficient}, where the authors  have proposed the  adaptive perturbation based algorithm (APERT, a pre-processing algorithm) using principal component analysis (PCA), two-timescale stochastic approximation and sequential probability ratio test (SPRT~\cite{poor2013introduction}) to distinguish between clean and adversarial images.

\subsection{Our Contributions}
We have made following contributions in this paper:
\begin{itemize}
    \item We propose a novel compressive sensing based adaptive defence (CAD) algorithm to combat $l_0$, $l_2$, $l_\infty$ norm attacks as well as gradient based attacks, with much lower computational complexity compared to existing works. The computational complexity is $\mathcal{O}(N^2)$ where $N$ is the number of pixels in an image. 
    \item CAD is the first algorithm that can detect the type of attack if it falls within certain categories (such as $l_0, l_2, l_{\infty}$), and choose an appropriate classification algorithm to apply on the potentially adversarial image. To this end, we have adapted the popular exponential weight algorithms~\cite{auer1995gambling},~\cite{chafaa2020exploiting} from multi-armed bandit literature to our setting, which adaptively assigns a score to each attack type, thus guiding us in choosing the  appropriate recovery algorithm (e.g., CoSaMP, basis pursuit etc.). The CAD  algorithm does not require any prior knowledge of the adversary.
     \item We consider adversarial perturbation in the frequency domain instead of the time domain while formulating the problem, which allows us to counter both low as well as high frequency spectral components.
    \item We propose modified basis pursuit using a {\em novel} constraint to mitigate $l_0$ and $l_\infty$ norm attacks, and establish its performance bound.
    \item Our work has the potential to trigger a new line of research where compressive sensing and multi-armed bandits can be used for detection and classification of adversarial videos.
\end{itemize}
\subsection{Organization}
This paper is further arranged as follows.  Description of various recovery algorithms and their  performance bounds are  established in Section~\ref{section:proposed method}. The proposed CAD algorithm is described in Section~\ref{section:Algorithm}. Complexity analysis of CAD  is provided in Section~\ref{section:Complexity Analysis}, followed by the numerical results  in Section~\ref{section: Experiments} and conclusions  in Section~\ref{section:Conclusion}.

\section{Basic model and various Recovery algorithms}\label{section:proposed method}
In this section, we define the basic problem and propose various recovery algorithms assuming that the attack type is known to the classifier. It is noteworthy that here we propose modified versions of basis pursuit to combat $l_0$ and $l_{\infty}$ attacks in the spectral domain, and provide performance bounds for these algorithms. The background theory provided in this section are prerequisites to understand the performance of the proposed CAD algorithm later under various circumstances.

\subsection{Problem Formulation}\label{subsection:problem formulation}
Let us consider a clean, vectorized image $\bm{x}\in \mathbb{R}^{N \times 1}$, and let us assume that it is $k$-sparse \cite[Definition~$2.1$]{foucart2013invitation} in discrete Fourier transform domain. Let its Fourier coefficients be ${\hat{\bm{x}}}=\bm{Fx}$, where $\bm{F}\in \mathbb{C}^{N \times N}$ is the DFT matrix. The adversary modifies the image in spectral domain by adding an error vector $\bm{e}$ to $\hat{\bm{x}}$, and the  distorted image becomes $\bm{y}=\bm{F}^{-1}{(\hat{\bm{x}}}+\bm{e})$. For $l_0$ attack,   $\bm{e}$ is assumed to be $\tau$-sparse, so that   ${\hat{\bm{x}}}+\bm{e}$ becomes at most $(k+\tau)$ sparse in Fourier domain. 
Defining $\bm{A} \doteq \bm{F^{-1}}$ and $\bm{\beta} \doteq \bm{F^{-1}e}$, the modified image becomes $\bm{y}=\bm{A}{\hat{\bm{x}}}+\bm{\beta}$. Our objective is to find $\hat{\bm{x}}$ from $\bm{y}$. We will solve this problem iteratively by  using compressive sensing based adaptive defense (CAD) algorithm comprising compressive sampling matching pursuit (CoSaMP) and  modified version of basis pursuit for various attacks, and an adapted version of the exponential weight algorithm for selecting the recovery algorithm.

\subsection{Compressive Sampling Matching Pursuit, CoSaMP}
We know that images are compressible signals as their coefficients decay  rapidly in Fourier domain when arranged according to their magnitudes. CoSaMP~\cite{needell2009cosamp} iteratively recovers the approximate Fourier coefficients of a compressible signal from noisy samples given that the signal is sparse in the Fourier domain; it is based on orthogonal matching pursuit (OMP), but provides  stronger guarantee than OMP.  The authors of \cite{needell2009cosamp} have shown that this algorithm produces a $2k$-sparse recovered vector whose recovery  error in $L_2$ norm is comparable with the scaled approximation error in $L_1$ norm.  CoSaMP provides optimal error guarantee for sparse signal, compressible signal and arbitrary signal.

Since we do not know apriori whether the attack is $l_0, l_2, l_\infty$ or gradient-based, and since it is difficult to infer the type of the attack initially,   we use CoSaMP along with various versions of basis pursuit for Fourier coefficient recovery. This is further motivated by the fact that CoSaMP is robust against arbitrary injected error~\cite{needell2009cosamp}. However, our proposed CAD algorithm (described in Section~\ref{section:Algorithm})  also adaptively assigns a score to each recovery scheme   via the exponential weight algorithm using the residue-based feedback for each algorithm, and probabilistically selects an algorithm in each iteration based on the assigned scores. The exponential weight algorithm is typically used to solve online learning problems that involve exploration and exploitation, and the robustness of CoSaMP facilitates exploration especially at the initial phase when the algorithm has not developed a strong belief about the type of attack.  In this connection, it is worth mentioning that CoSaMP has provably strong performance bounds in all cases and also works well for highly sparse signals. 

Let us denote by ${\hat{\bm{x}}}^0$  the initialisation before applying CoSaMP algorithm (usually we take ${\hat{\bm{x}}}^0=0$). The quantity ${\hat{\bm{x}}}_{h(k)}$ is a $k$-sparse vector (i.e., its $l_0$ norm is at most $k$) that consists of $k$ largest entries (in terms of absolute values) of ${\hat{\bm{x}}}$. We also define  ${\hat{\bm{x}}}_{t(k)}={\hat{\bm{x}}}-{\hat{\bm{x}}}_{h(k)}$. The iteration number in the CoSaMP algorithm is denoted by $n$.

The performance guarantee of CoSaMP is provided through the following theorem:
\begin{theorem} \label{theorem:Cosamp-bound}
Suppose that the $4k^{th}$ restricted isometry constant of the matrix $\bm{A}\in \mathbb{C}^{N \times N}$ satisfies $\delta_{4k}<0.47$. Then, for ${\hat{\bm{x}}}\in \mathbb{C}^N$, $\bm{\beta} \in \mathbb{C}^N$, and $S \subset[N]$ with card(S) = $k$, the Fourier coefficients ${\hat{\bm{x}}}^n$ defined by CoSaMP with $\bm{y}= \bm{A}{\hat{\bm{x}}} +\bm{\beta}$ satisfies:
\begin{equation}\label{equation:l2 norm cosamp}
\|{\hat{\bm{x}}}^n - {\hat{\bm{x}}}_{h(k)}\|_2 \leq \rho^n\|{\hat{\bm{x}}}^0-{\hat{\bm{x}}}_{h(k)}\|_2 + \tau \|\bm{A}{\hat{\bm{x}}}_{t(k)} + \bm{\beta}\|_2
\end{equation}
where the constant $0 < \rho <1$ and $\tau > 0$ depend only on $\delta_{4k}$.
\end{theorem}
\begin{proof}
The proof  is similar to that of \cite[Theorem~$6.27$]{foucart2013invitation}.
\end{proof}

 \subsection{Combating $l_2$ Attack using basis pursuit}\label{subsection: l2 basis pursuit}
 Standard basis pursuit is chosen to counter $l_2$ perturbation~\cite{donoho2006stability} since it minimizes the $l_1$ norm of Fourier coefficients while constraining the $l_2$-norm of the injected error. Let us assume that the $l_2$ perturbation satisfies $||\bm{F}^{-1}\bm{e}||_2\leq \eta$ for a small $\eta$, and hence is  imperceptible to human eye. Since $\bm{F^{-1}}$ is an orthonormal matrix, we can write it as $||\bm{e}||_2\leq \eta$.

Let $\sigma_k({\hat{\bm{x}}})_1 \doteq \min_{\mathbf{||\bm{z}||_0\leq k}} ||{\hat{\bm{x}}}-\bm{z}||_1$. Performance bound for the standard basis pursuit algorithm is provided in the following theorem: 
 \begin{theorem} \label{theorem:basispursuit-bound l2}
Suppose that the $2k^{th}$ restricted isometry constant of the matrix $\bm{A}\in \mathbb{C}^{N \times N}$ satisfies $\delta_{2k}<0.624$. Then, for any ${\hat{\bm{x}}}\in \mathbb{C}^N$ and $\bm{y} \in \mathbb{C}^N$ with $||\bm{A}{\hat{\bm{x}}}-\bm{y}||_2\leq {\eta}$, a solution ${\hat{\bm{x}}}^{*}$ of $\min_{\mathbf{\bm{z}\in C^N}} ||\bm{z}||_1$ subject to $||\bm{Az}-\bm{y}||_2\leq {\eta}$ approximates the ${\hat{\bm{x}}}$ with errors
\begin{equation}\label{equation:l1 norm basis pursuit l2 attack}
\|{\hat{\bm{x}}} - {\hat{\bm{x}}}^*\|_1 \leq  C\sigma_k({\hat{\bm{x}}})_1 + D\sqrt{k}\eta
\end{equation}
\begin{equation}\label{equation:l2 norm basis pursuit l2 attack}
\|{\hat{\bm{x}}} - {\hat{\bm{x}}}^*\|_2 \leq  \frac{C}{\sqrt{k}}\sigma_k({\hat{\bm{x}}})_1 + D\eta
\end{equation}
where the constants $C,D > 0$ depend only on $\delta_{2k}$.
\end{theorem}
\begin{proof}
The proof is similar to \cite[Theorem~$6.12$]{foucart2013invitation}
\end{proof}
 From  Theorem~\ref{theorem:basispursuit-bound l2}, it is clear that, in order to guarantee unique recovery of largest $k$  Fourier coefficients, sensing matrix $\bm{A}$ should satisfy restricted isometry property (RIP)~\cite[Definition~$6.1$]{foucart2013invitation} of order $2k$. It has been observed that with high probability, random Gaussian and partial Fourier matrices satisfy RIP properties~\cite{cheraghchi2013restricted}, which ensures that any $2k$ columns in matrix $\bm{A}$ are linearly independent. 
 We can relate performance bound  \eqref{equation:l2 norm basis pursuit l2 attack} in spectral domain with that in time domain, since $\bm{F^{-1}}$ is an orthonormal matrix.

\subsection{Combating $l_0$ Attack using basis pursuit}
In Section~\ref{section:Algorithm}, we employ another modified version of basis pursuit  to counter $l_0$ attack; this involves a slightly different formulation. Let us assume that the perturbation error $\bm{e}$ is  $\tau$ sparse, and let us arrange perturbations of error vector $\bm{e}$ in ascending order $[e_1, e_2,...e_\tau...0]$. In $l_0$ attack, the attacker has constraints only on the number of Fourier coefficients that can be perturbed. Since according to the uncertainty principal~\cite{fefferman1983uncertainty} any image cannot be simultaneously narrow in the pixel domain as well as in spectral domain,  the $l_{\infty}$ norm of the injected error $\bm{e}$ under $l_0$ attack should have small enough to remain imperceptible to the human eye, i.e.,  $|\bm{e}|_{\infty} < {\eta^{'}}$, for some constant ${\eta^{'}}$. Now, it is well known that $\|\bm{e}\|_2   \leq  \|\bm{e}\|_1$, and we also notice that 
$\|\bm{e}\|_1 = |e_1| + |e_2| +...+ |e_\tau| \leq \tau|e_\tau|$, which yield $\|\bm{e}\|_2 \leq \tau|e_\tau| \leq \tau{\eta^{'}}$.

The performance bound for the modified basis pursuit algorithm under $l_0$ attack is provided in the following theorem:
\begin{theorem} \label{theorem:basispursuit-bound l0 attack}
Suppose that the $2k^{th}$ restricted isometry constant of the matrix $\bm{A}\in \mathbb{C}^{N \times N}$ satisfies $\delta_{2k}<0.624$. Then, for any ${\hat{\bm{x}}}\in \mathbb{C}^N$ and $\bm{y} \in \mathbb{C}^N$ with $||\bm{A}{\hat{\bm{x}}}-\bm{y}||_2\leq {\tau{\eta^{'}}}$, a solution ${\hat{\bm{x}}}^{*}$ of $\min_{\mathbf{\bm{z}\in C^N}} ||\bm{z}||_1$ subject to $||\bm{Az}-\bm{y}||_2\leq {\tau{\eta^{'}}}$ approximates the ${\hat{\bm{x}}}$ with errors
\begin{equation}\label{equation:l1 norm basis pursuit}
\|{\hat{\bm{x}}} - {\hat{\bm{x}}}^*\|_1 \leq  C\sigma_k({\hat{\bm{x}}})_1 + D\sqrt{k}\tau{\eta^{'}}
\end{equation}
\begin{equation}\label{equation:l2 norm basis pursuit}
\|{\hat{\bm{x}}} - {\hat{\bm{x}}}^*\|_2 \leq  \frac{C}{\sqrt{k}}\sigma_k({\hat{\bm{x}}})_1 + D\tau{\eta^{'}}
\end{equation}
where the constants $C,D > 0$ depend only on $\delta_{2k}$.
\end{theorem}
\begin{proof}
This Theorem can be followed easily using Theorem~\ref{theorem:basispursuit-bound l2} and the fact that $\|\bm{e}\|_2 \leq \tau|e_\tau| \leq \tau{\eta^{'}}$ as discussed earlier.
\end{proof}

 \subsection{Combating $l_\infty$ Attack using basis pursuit}
 Let us assume that $||\bm{e}||_\infty < \eta''$. Now, since $\bm{F}$ is orthonormal, 
 \begin{equation}\label{equation:steps Relation l2 and l_infinity}
||\bm{F}^{-1}\bm{e}||^2_2= ||\bm{e}||_2^2 \leq N\max_i(|e_i|^2)=N||\bm{e}||^2_\infty
\end{equation}
and hence 
\begin{equation}\label{equation:Relation l2 and l_infinity}
||\bm{e}||_2 \leq \sqrt{N}||\bm{e}||_\infty \leq \sqrt{N}\eta''
\end{equation}
The performance guarantee for    modified basis pursuit under $l_\infty$ attack is provided in the following theorem:
\begin{theorem}\label{theorem:Modified basis pursuit l_infinity}
Suppose that the $2k^{th}$ restricted isometry constant of the matrix $\bm{A}\in \mathbb{C}^{N \times N}$ satisfies $\delta_{2k}<0.624$. Then, for any ${\hat{\bm{x}}}\in \mathbb{C}^N$ and $\bm{y} \in \mathbb{C}^N$  with $||\bm{A}{\hat{\bm{x}}}-\bm{y}||_2\leq \sqrt{N}{\eta''}$, a solution ${\hat{\bm{x}}}^{*}$ of $\min_{\mathbf{\bm{z}\in C^N}} ||\bm{z}||_1$ subject to $||\bm{Az}-\bm{y}||_2\leq \sqrt{N}{\eta''}$ approximates the ${\hat{\bm{x}}}$ with errors
\begin{equation}\label{equation:l1 norm modified basis pursuit}
{\|{\hat{\bm{x}}} - {\hat{\bm{x}}}^*\|}_1 \leq  C\sigma_k({\hat{\bm{x}}})_1 + D\sqrt{kN}\eta''
\end{equation}
\begin{equation}\label{equation:l2 norm modified basis pursuit}
\|{\hat{\bm{x}}} - {\hat{\bm{x}}}^*\|_2 \leq  \frac{C}{\sqrt{k}}\sigma_k({\hat{\bm{x}}})_1 + D\sqrt{N}\eta''
\end{equation}
where the constants $C,D > 0$ depend only on $\delta_{2k}$.
\end{theorem}
\begin{proof}
The proof follows easily from Theorem~\ref{theorem:basispursuit-bound l2} and (\ref{equation:Relation l2 and l_infinity}).
\end{proof}

\subsection{Combating $l_1$ attack using basis pursuit}
If $\bm{e}$ is such that $\|\bm{e}\|_1 < \eta$, then the error in the recovered image also satisfies $
\|\bm{F}^{-1}\bm{e}\|_2 = \|\bm{e}\|_2 \leq \|\bm{e}\|_1 \leq \eta$, and we can solve the same $l_1$ minimization problem with the same constraint as in Section~\ref{subsection: l2 basis pursuit} for $l_2$ attack. Similarly, its performance bound will be given by Theorem~\ref{theorem:basispursuit-bound l2}.

\section{The Compressive Sensing   based Adaptive Defense (CAD) algorithm}\label{section:Algorithm}
In this section, we propose our main algorithm to combat adversarial images. Since the CAD algorithm does not have any  prior knowledge on the type of attack, CAD algorithm  employs an adaptive version of the exponential weight algorithm~\cite{auer1995gambling},~\cite{chafaa2020exploiting} for exploration and exploitation  to assign a score on each possible attack type, and  chooses an appropriate recovery method based on the inferred nature of the injected error. In this paper, we consider four actions, i.e., four different ways to recover $k$-sparse Fourier coefficients, corresponding to   different types of perturbation:
\begin{itemize}
    \item {\bf CoSaMP (Action 1):} This greedy approach allow us to accurately approximate the Fourier coefficients initially when we do not have any belief for the type of attack. As iterations progress,   the algorithm explores other actions as well. 
    
    \item {\bf {Modified Basis pursuit \bm{$L_0$} (Action 2): }} A modified form of basis pursuit with novel constraint $\|\bm{e}\|_2 \leq \tau{\eta^{'}}$ is used to tackle $l_0$ perturbation attack.
    \item {\bf  Standard Basis pursuit \bm{$L_1$} and \bm{$L_2$} (Action 3): } Standard basis pursuit method is used to counter both $l_1$ and $l_2$ attack.
    \item {\bf  Modified Basis pursuit \bm{$L_\infty$} (Action 4): } Modified basis pursuit is used to tackle $l_\infty$ norm attack, using novel constraint given by (\ref{equation:Relation l2 and l_infinity}).
\end{itemize}

In the next three subsections, we discuss three major aspects of our proposed CAD algorithm: (i) adaptive exponential weight algorithm for choosing an appropriate recovery scheme, (ii) actions and feedback, and (iii) stopping criteria. 
\subsection{The adaptive version of exponential weight for choosing the recovery scheme}
Algorithm 1 summarizes the overall defence strategy. In each iteration $t$, the algorithm chooses randomly an action using a probability distribution $p_{a_i}(t)$ where  $a_i, i \in \{1,2,3,4\}$ denotes the action chosen. 

The probability of choosing an action is given by exponential weighting:  
\begin{equation}\label{equation:probability distribution}
p_{a_i}(t) = (1 - \gamma)\frac{\exp{(\sigma S_{a_i}(t - 1))}}{\sum_{m = 1}^4 \exp{(\sigma S_{a_m}(t - 1))}} + \frac{\gamma}{4}
\end{equation}
where $S_{a_i}(t - 1) = \sum_{\tau = 1}^{t - 1}r_{a_i}(\tau)$ is the total score up for the action $a_i$. Here $\sigma$ and $\gamma$ are tuning parameters such that $\sigma > 0$ and $\gamma \in (0,1)$.  The reward for action $a_i$ at the $t$-th iteration,  $r_{a_i}(t)$ is the following:
\begin{equation}\label{equation:reward}
r_{a_i}(t)= 
\begin{cases}
    \frac{\lambda}{p_{a_i}(t)},& \text{if } f_{a_i}(t) = 1\\
    \frac{-1}{1 - p_{a_i}(t)},   & \text{if } f_{a_i}(t) = 0\\
    0 &\text{if } a_i \text{ is not chosen in the $t$-th iteration}
\end{cases}
\end{equation}
Here $f_{a_i}(t)$ is a binary feedback that is obtained by checking certain conditions for action $a_i$  in the $t$-th iteration; this feedback signifies the applicability of action $a_i$. If action $a_i$ is chosen in the $t$-th iteration and if its feedback $f_{a_i}(t)=1$, the actual reward $\lambda>0$  is divided by $p_{a_i}(t)$ so that an  unbiased estimate of the reward is obtained. On the other hand, if $f_{a_i}(t)=0$, then a penalty of $-1$ is assigned for $a_i$. However, this penalty is divided by $(1-p_{a_i}(t))$ to ensure that, if $p_{a_i}(t)$ is small because it has not been chosen frequently earlier, the penalty incurred by $a_i$ in the $t$-th iteration remains small.

The action in each iteration is chosen   in the following way. With probability $\gamma$, one action is randomly chosen from uniform distribution. This is done to ensure sufficient exploration of all recovery algorithms irrespective of the reward accrued by them at the initial phase. On the other hand, with probability $(1-\gamma)$, each action is chosen randomly with a probability depending on its accumulated score.

\subsection{Detailed discussion on actions and feedback}
In action-1 CosaMP, the following steps are involved:
\begin{itemize}
    \item {\bf Identification: } Steps 1 and 2 provide the signal proxy for the residual error vector and find out the indices of largest $2k$ entries.
    \item {\bf Support Merger: } Step 3 merges the set of new indices with set of indices of current Fourier coefficients approximation.
    \item {\bf Estimation: } Step 4 computes the least squares to obtain the approximate Fourier coefficients on merged set $R$.
    \item {\bf Pruning: } Steps 5 and 6 maintain only  largest $k$ Fourier coefficients obtained from least square approximation.
\end{itemize}

Details of each action 2,3 and 4 are mentioned in the algorithm.

\begin{algorithm}
\caption{CAD algorithm}
{\bf Input:} The measurement matrix $\bm{A}={\bm{F}^{-1}}$, test image vector $\bm{y}$, dimension of image vector $N$, sparsity parameters  $\tau$ and $k$, perturbation levels $\eta$, $\eta'$ and $\eta''$, Mahalanobis Distance (MD) threshold $\theta$, stopping time $T$, stopping time threshold parameters $\Delta$ and $\delta$, and also  $\alpha$, $\beta$, $m$, $\gamma \in (0,1)$, $\lambda>0$ $ \sigma > 0$.
 
 {\bf Initialisation:} Set Cumulative score $S_{a_i}(0) = 0 \forall i \in \{1,2,3,4\}$, Fourier coefficients ${\hat{\bm{x}}}^0 = 0$, residual error $\bm{v}^0 = y$ and $p_{a_i}(1) = 1/4$ for all actions in $\mathcal{A} = \{ a_1, a_2, a_3, a_4 \}$ \\
 
 \KwResult {${\hat{\bm{x}}}$ which is $k$ sparse approximation of Fourier coefficients}
 {\bf Actions:}
 
  \begin{itemize}
  
  \item {\boldmath$a_1$:} {\bf Action 1} \\
  \begin{enumerate}
  \item $\bm{z} \leftarrow \bm{A}^{*}\bm{v}^{t-1}$ 
  \item $\Omega \leftarrow supp(\bm{z}_{2k})$   
 \item $R \leftarrow \Omega \cup supp({\hat{\bm{x}}}^{t-1})$ 
 \item $\bm{b}_{|R} \leftarrow \bm{A}^{\dagger}_R \bm{y}$   
 \item $\bm{b}_{|R^c} \leftarrow 0$  
 \item   {\bf Return:} ${\hat{\bm{x}}}^t \leftarrow \bm{b}_k$ 
 \end{enumerate}
   \item {\boldmath$a_2$:} {\bf Action 2} \\
 {\bf Return:} ${\hat{\bm{x}}}^{t} \leftarrow \mathop {\arg \min }\limits_{\bm{z} \in \mathbf{C}^N} \|\bm{z}\|_1$ \bf{s.t.} $\|\bm{Az}-\bm{y}\|_2 < \tau\eta^{'}$ \\
 
  \item {\boldmath$a_3$:} {\bf Action 3} \\
 {\bf Return:} ${\hat{\bm{x}}}^{t}\leftarrow\mathop {\arg \min }\limits_{\bm{z} \in \mathbf{C}^N} \|\bm{z}\|_1$ \bf{s.t.} $\|\bm{Az}-\bm{y}\|_2 < \eta$\\
 
 \item {\boldmath$a_4$:} {\bf Action 4} \\
{\bf Return:} ${\hat{\bm{x}}}^{t}\leftarrow\mathop {\arg \min }\limits_{\bm{z} \in \mathbf{C}^N} \|\bm{z}\|_1$ \bf{s.t.} $\|\bm{Az}-\bm{y}\|_2 < \sqrt{N}\eta^{''}$\\
 \end{itemize}

 \For{ $t = 1,...,T$} {
 \begin{enumerate}
     \item  Select action $a_i$, $i \in \{1,2,3,4\}$ with sampling distribution $p_{a_i}(t)$ using (\ref{equation:probability distribution}).
     \item  Perform some more number of initial iterations of chosen action $a_i$ compared to the last time when $a_i$ was chosen.  
     \item Find top $k$ Fourier coefficients i.e. ${\hat{\bm{x}}}^t = \hat{\bm{x}}_{h(k)}$ using the output in the previous step.
     \item Calculate the residual error $\bm{v}_t \leftarrow \bm{y}-\bm{A}{{\hat{\bm{x}}}^{t}}$.
     \item Feedback $f_{a_i}(t)=1$ is set if following condition holds for the chosen action:
     \begin{itemize}
        \item {\boldmath$a_1$:} $\|\bm{v}_t\|_2 < \alpha ${ \bf Or }$ MD < \theta $ {\bf And} $\|\bm{v}_t\|_\infty < m $
         \item {\boldmath$a_2$:} $\|\bm{v}_t\|_2 > \alpha $ {\bf And }$\|\bm{v}_t\|_0 < \tau $
         \item {\boldmath$a_3$:} $\|\bm{v}_t\|_2 > \alpha $ {\bf And} $m < \|\bm{v}_t\|_\infty < \beta $
         \item {\boldmath$a_4$:} $\|\bm{v}_t\|_2 > \alpha $ {\bf And } $\|\bm{v}_t\|_\infty > \beta$
     \end{itemize}
    \item Calculate reward $r_{a_i}(t)$ using (\ref{equation:reward}).
    \item Update cumulative score
    \begin{itemize}
        \item $S_a(t) = S_a(t-1) + r_{a_i}(t), a = a_i $
        \item $S_a(t) = S_a(t-1), \forall a \neq a_i$
    \end{itemize}

     \item \If {$p_{a_i}(t) > \Delta${ \bf Or} $\|\bm{v}_t\|_2 < \delta$}
    {
        { \bf break}
    }
    \end{enumerate}
    
 }
{\bf Recovery method chosen} = $ \mathop {\arg \max }\limits_{a} S_a(T)$\\
\If {$ \mathop {\max }\limits_{a} S_a(T) \leq 0$}
    {
        {\bf Recovery method chosen} = CoSaMP
    }
\end{algorithm}

We choose the following feedback criterion i.e. $f_{a_i} = 1 $ for each action:
\begin{itemize}
    \item{\bf Action 1:} It is quite intuitive that if there is no attack then the $l_2$ norm of the residual error will be upper bounded by just recovery error at the end of the algorithm. Hence, we set its upper bound equal to  the parameter $\alpha$. The maximum absolute value in the residual vector is upper bounded by the parameter $m$. If these inequalities are satisfied in each iteration, then the algorithm concludes that there is no attack, hence $f_{a_i}=1$. We can also calculate the Mahalanobis distance (MD)~\cite{de2000mahalanobis} using (\ref{equation:mahalanobis}), between the residual error  of a test image and that of  the clean images. This is used as another alternative criterion to determine whether the image is malicious or not by comparing with some threshold parameter $\theta$.  
    \begin{equation}\label{equation:mahalanobis}
        MD^2 = (\bm{v} - \hat{\bm{m}})^T \bm{C}^{-1} (\bm{v} - \hat{\bm{m}})
    \end{equation}
    where $\bm{v}$ is the residual error of test image and $\hat{\bm{m}}$ and $\bm{C}$ is the mean and covariance of residual error of clean images respectively. This is reminiscent of  the popular $\chi^2$ detector used in anomaly detection.
    \item{\bf Action 2:} Under $l_0$ attack, the number of non-zero entries in its perturbation vector should be upper bounded by some parameter $\tau$. Hence, we use the conditions $\|\bm{v}_t\|_2 > \alpha $  and $\|\bm{v}_t\|_0 < \tau $. This can be explained from the fact that $\bm{v}_t$ includes perturbation error along with recovery error.
    \item{\bf Action 3:} Along with the previous   condition $\|\bm{v}_t\|_2 > \alpha $, here we assume that maximum absolute perturbation in case of $l_2$ or $l_1$ attack is upper bounded by  $\beta$ and lower bounded by     $m$.
    \item{\bf Action 4:} Checking for $l_\infty$ attack additionally requires us to verify whether the maximum residual error component which acts as a proxy for the maximum perturbation is greater than  $\beta$.
\end{itemize}
Choosing an action yields a feedback status which influences the reward values as in \eqref{equation:reward} and consequently the probabilities of choosing all actions.

We numerically observed, in addition to the above feedback criteria, that the residual vector contains a   large number of nonzero entries for actions 3 and 4 for adversarial grayscale images such as the MNIST dataset. Hence,   in our experiments in Section~\ref{section: Experiments}, we additionally check whether $\Vert \bm{v}_t \Vert_0$ is above a threshold.

\subsection{Stopping criteria}
CAD can be run till the maximum limit $T$ for the number of iterations is reached. However, if either of the two conditions $p_{a_i}(t) > \Delta$ for some $i \in \{1,2,3,4\}$ and $\|\bm{v}_t\|_2 < \delta$ is met before that for two given threshold parameters $\Delta$ and $\delta$, then the iteration will stop. The condition  $p_{a_i}(t) > \Delta$ means that it is optimal to choose action $i$ with high probability, and hence no further exploration is required. The condition $\|\bm{v}_t\|_2 < \delta$ means that most likely the test image is clean, and hence there is no need to investigate it further.

At the end, the appropriate recovery method is chosen according to  the action   which achieves the maximum cumulative score. However, if the maximum score is negative at this time, then it implies that CAD is unable to clearly identify the type of attack, and hence   CoSaMP is chosen as a default recovery method due to its robustness.

\section{Complexity Analysis}\label{section:Complexity Analysis}
CoSaMP has the following five steps:  forming signal proxy, identification, support merger, least square estimation and  pruning. The sensing matrix $\bm{A} = \bm{F}^{-1}$ has   dimension $N \times N$ and   sparsity $k$. Hence, following standard matrix vector multiplication,  time complexity for each of the five  steps~\cite{needell2009cosamp} are obtained as  $\mathcal{O}(N^2)$, $\mathcal{O}(N)$, $\mathcal{O}(k)$, $\mathcal{O}(kN)$, $\mathcal{O}(k)$ respectively. Hence, CoSaMP has time complexity   $\mathcal{O}(N^2)$ for each  iteration~$t$. For actions 2,3 and 4, we need to  solve $l_1$ minimization problem with different constraints, which can be solved efficiently by a standard convex optimization solver in polynomial time $\mathcal{O}(p(N))$.

For any action, choosing the top $k$ Fourier coefficients is similar to the CoSaMP pruning step, and it can be done by a  sorting algorithm   in $\mathcal{O}(k \log k)$ time.  The number of operations required to calculate the residual error $\bm{v}_t=\bm{y}- \bm{A} \hat{\bm{x}}^t$ for  a  $k$ sparse vector $\bm{\hat{x}}^t$ is $\mathcal{O}(kN)$. Calculating various  norms such as $l_2$, $l_0$ and $l_\infty$ require $\mathcal{O}(N)$ each time. Also, the number of iterations is upper bounded by $T$. 

Hence, the overall computational complexity of CAD will be of $\mathcal{O}(TN^2+ T p(N))$.

\section{Experiments}\label{section: Experiments}
We conducted our experiments on MNIST~\cite{lecun1998mnist} and CIFAR-10~\cite{krizhevsky2009learning} data sets for pixels lying in between $[0,1]$. Discrete Cosine Transform (DCT) domain is used in experiments to get sparse coefficients. We consider only white box attacks since the attacker in a black box attack has access to much less information than a white box attacker, and hence is less effective in general. All  experiments were performed in Google Colab.
\subsection{Attack setup}
 Foolbox~\cite{rauber2017foolbox} is an open source library available in python that can exploit the vulnerabilities of DNNs and  generate  various malicious attacks. All our evaluations are done using the 2.3.0 version of Foolbox library. We evaluate our compressive sensing based adaptive defense (CAD) against five major state-of-the-art white box adversarial attacks. They are projected gradient descent (PGD)~\cite{madry2017towards}, basic iterative method (BIM)~\cite{kurakin2016adversarial}, fast gradient sign method (FGSM)~\cite{goodfellow2014explaining}, Carlini Wagner(L2)(CW) attack~\cite{carlini2017towards} and Jacobian saliency map attack (JSMA)~\cite{papernot2016limitations}. In the PGD attack, 40 iterations steps with random start are used in Foolbox. For the C\&W attack, we use 10,000 iteration steps with a learning rate of 0.01. In the BIM attack, the number of iterations is set to 10 and limit on perturbation size is set to 0.3. We use the default parameters  of foolbox library for the FGSM attack. In JSMA attack maximum iteration is set to 2000 and perturbation size in $l_0$ norm is set to 20 and 35 for MNIST and CIFAR-10 respectively. All attacks used in this are bounded under $l_\infty$ norm with perturbation size $\epsilon = 0.3$ and $\epsilon = 8/255$ for MNIST and CIFAR-10 respectively.

As the authors of~\cite{sharma2019effectiveness} observed that data sets  such as MNIST ($28 \times 28$) and CIFAR-10 ($32 \times 32$) are  too low dimensional to exhibit a diverse frequency spectrum. Hence, we do not test our algorithm against low frequency adversarial perturbation attacks.
\subsection{Training and testing setup}
For training, we use clean, compressed, reconstructed images using only top $k$ DCT coefficients. Then we test the DNN based classifier  against perturbed images (without any  reconstruction) and note down its adversarial accuracy and loss.   Then we employ our proposed CAD algorithm to reconstruct the adversarial images to obtain corrected classification accuracy and loss for each attack.

The model architecture used for MNIST is described in Table~\ref{table Mnist model architecture}. We use an RMSprop optimizer in Keras with cross-entropy loss for MNIST. For CIFAR-10 we use ResNet (32 Layers)~\cite{targ2016resnet} model having Adam optimizer with cross-entropy loss having batch size = 128 and epoch = 50. We randomly choose 7000 and 2050 images for MNIST and CIFAR-10 respectively from the training set, and train the classifier with its reconstructed and compressed (reconstructed by taking top $k$ DCT coefficients of the image) images. In MNIST, we take 1000 corrupted images randomly from the test set for each attack. Since 3 channels are available in CIFAR-10,  attacks are much expensive to execute and time complexity is  $\mathcal{O}(3TN^2+ 3T p(N))$. Hence, we choose only 250 images randomly from the test set for each attack to evaluate our CAD algorithm.

\begin{table}[ht!]
    \caption{MNIST model architecture with batch size = 128, epochs = 15}
    \begin{tabular}{ |c|c|c| } 
        \hline 
        Layer & Type & Properties\\
        \hline
         1 & Convolution & 32 channels, 3x3 Kernel, activation=Relu\\ 
         2 & Convolution & 64 channels, 3x3 Kernel, activation=Relu\\ 
         3 & Convolution & 64 channels, 3x3 Kernel, activation=Relu\\ 
         4 & Max pooling & 2x2, Dropout = 0.25, activation=Relu \\ 
         5 & Fully connected & 128 neurons, activation=Relu \\
         6 & Fully connected & 64 neurons, activation=Relu, Dropout = 0.5\\
         7 & Fully connected & 10 neurons, activation = Softmax\\
        \hline
    \end{tabular}
    \label{table Mnist model architecture}
\end{table}

Various parameters used in the algorithm are as follows:
\begin{itemize}
    \item {\bf MNIST: } $k = 80$, $\alpha = 8$, $\beta = 5$, $m = 1.8$, $\tau = 15$, $\theta = 65$, $\gamma = 0.07$, $\sigma = 1.01$, $\lambda = 1.25$, different perturbation levels $\eta = 0.3$, ${\eta^{'}} = 0.15$ and ${\eta^{''}} = 0.04$.
    \item {\bf CIFAR-10: } $k = 300$, $\alpha = 7$, $\beta = 2.8$, $m = 1.4$, $\tau = 35$, Mahalanobis distance threshold~\cite{de2000mahalanobis} parameter $\theta$ is equal to 3.3, 3, 3.2 respectively for each channel, $\gamma = 0.45$, $\sigma = 1.0$, $\lambda = 5$, different perturbation levels $\eta = 0.5$, ${\eta^{'}} = 0.05$ and ${\eta^{''}} = 0.05$. 
\end{itemize}

 Stopping criterion parameters are $\Delta = 0.8$ and $\delta = 2$. For  action-2, although we mention in the algorithm that the number of non-zero entries should be less than  $\tau$ to satisfy the feedback condition,  there will be some recovery error components in practice. Hence, we count the number of entries greater than  a threshold 0.5 in the residual error vector, instead of exactly counting the number of non-zero entries. In CIFAR-10, we run our algorithm channel-wise to obtain reconstructed coefficients.

\subsection{Clean data accuracy and reverse engineering attack}
It has been observed that defenses that employ adversarial training suffer from the problem of clean data accuracy, i.e., the classifiers trained with adversarial images perform poorly for clean images. In order to address this problem, we evaluate the cross-entropy loss and classification accuracy of our algorithm on 10,000 uncompressed, clean test images in the first row of Table~\ref{table classification MNIST} and~\ref{table classification cifar10} for both MNIST and CIFAR-10. Our results show  that our trained model using reconstructed and compressed images works effectively in classifying the uncompressed clean images, compared to the competing algorithms.

 Reverse engineering attacks allow the attacker to determine the decision rule by monitoring the output of the classifier for sufficient number of query images~\cite{tramer2016stealing}. CAD algorithm randomly chooses the recovery algorithm based on the nature of the recovery error, hence it is completely non-deterministic. In order to impart more uncertainty to the recovered coefficients and confuse the attacker, one can randomly initialize $\hat{\bm{x}}^{0}$ instead of initializing it with all zero vectors. Hence, generating a  reverse engineered attack   for our proposed CAD algorithm becomes difficult.

\subsection{Classification accuracy}
For comparison, we take 5 state-of-the-art defenses proposed in recent years:
\begin{enumerate}
    \item {\bf CRD~\cite{dhaliwal2020compressive}:} We choose CRD defense for comparison since it is also based on compressive sensing. The authors  of \cite{dhaliwal2020compressive} provided different recovery methods of Fourier coefficients for each norm attack, but did not test  their defense against gradient based attacks like FGSM, PGD since these two attacks do not yield    any norm condition.
     \item {\bf Madry et al. defense~\cite{madry2017towards}:}  It is min-max optimization based defense using adversarial training to combat adversarial attack. 
     \item {\bf Yu et al. defense~\cite{yu2018interpreting}:}  This  also is based on adversarial training, using decision surface geometry as parameter.
     \item {\bf CMT~\cite{wang2020defending}:}  This defense uses collaborative learning to increase the complexity in searching adversarial images for the attacker. It is applicable   for non-targeted, blackbox and greybox attacks.
     \item {\bf Bafna et al. defense~\cite{bafna2018thwarting}:} This defense is based on compressive sensing techniques and is only applicable for $l_0$ attack. 
     \end{enumerate}

The classification accuracy of the all defenses are    computed for  targeted white box adversarial images (since white box attack is more effective), except for the CMT~\cite{wang2020defending} defense scheme which is evaluated for  non-targeted black box adversarial images.   Experimental results for cross-entropy loss and classification accuracy of both adversarial and clean images for each attack are provided in Tables~\ref{table classification MNIST} and \ref{table classification cifar10}.  It is clear from the tabulated results that CAD algorithm is outperforming CRD except in CW(L2) attack in MNIST dataset where the performance is slightly worse.  The defenses proposed in \cite{madry2017towards} and \cite{yu2018interpreting} perform well for data sets having grayscale images (e.g., MNIST), but exhibit very    low classification accuracy for data sets having colored images (e.g., CIFAR-10). Finally, we compare CAD algorithm  with CMT \cite{wang2020defending}. Though white box attack usually performs  better than black box attack, proposed CAD algorithm against white box attack achieves much better classification accuracy compared to CMT under black box attack,  for the   MNIST data set. On the other hand, CAD algorithm against white box attack achieves comparable  classification accuracy compared to CMT under black box attack,  for the   CIFAR-10 data set.  It is to be noted that   CMT exhibits extremely poor classification accuracy against CW(L2) attack for both MNIST and CIFAR-10 data sets. Since PGD and BIM attacks are very similar, many of the defence papers demonstrate their performance against only one of PGD and BIM, as seen in Tables~\ref{table classification MNIST} and \ref{table classification cifar10}.

In Table~\ref{table l0 attack}, we compare the existing defenses against $l_0$ norm attack JSMA. It can be observed that CAD algorithm significantly outperforms others for MNIST. For CIFAR-10, CAD algorithm achieves high classification accuracy compared to  CRD~\cite{dhaliwal2020compressive} but poorer   accuracy compared to  CMT~\cite{wang2020defending}; however, one should remember that here CAD algorithm is evaluated against white box JSMA attack, while CMT is evaluated against black box JSMA attack. 

We also illustrate the reconstruction quality of randomly selected images (after performing inverse DCT on recovered coefficients) for each attack in Figures~\ref{Reconstruction quality of MNIST } \& \ref{Reconstruction quality of CIFAR-10 }. It is observed that the reconstructed images have high classification accuracy.

\begin{figure}
\centering
\begin{tabular}{ccccc}
\subfloat CW-L2 &
\subfloat FGSM &
\subfloat PGD &
\subfloat BIM &
\subfloat JSMA \\
{\includegraphics[width = 0.5in]{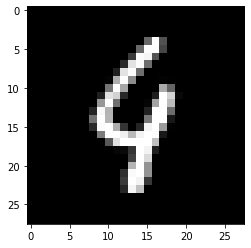}} &
{\includegraphics[width = 0.5in]{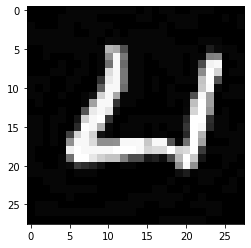}} &
{\includegraphics[width = 0.5in]{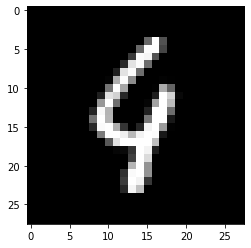}} &
{\includegraphics[width = 0.5in]{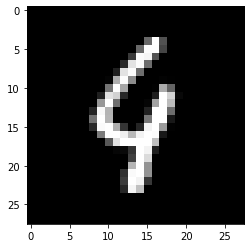}} &
{\includegraphics[width = 0.5in]{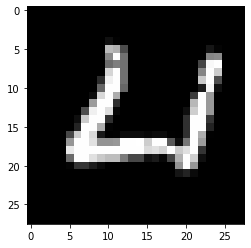}} \\
{\includegraphics[width = 0.5in]{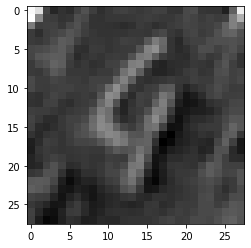}} &
{\includegraphics[width = 0.5in]{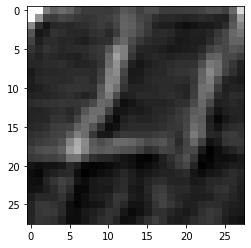}} &
{\includegraphics[width = 0.5in]{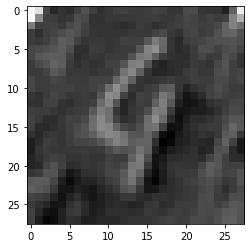}} &
{\includegraphics[width = 0.5in]{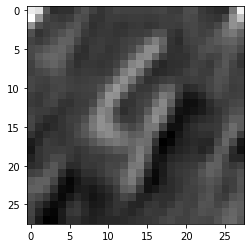}} &
{\includegraphics[width = 0.5in]{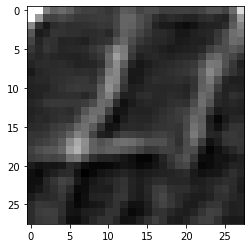}}\\
\end{tabular}
\caption{Reconstruction quality of MNIST images against various attacks. The first row shows the adversarial images and its reconstructed image is shown in second row.}
\label{Reconstruction quality of MNIST }
\end{figure}
\begin{figure}
\centering
\begin{tabular}{ccccc}
\subfloat CW-L2 &
\subfloat FGSM &
\subfloat PGD &
\subfloat BIM &
\subfloat JSMA \\
{\includegraphics[width = 0.5in]{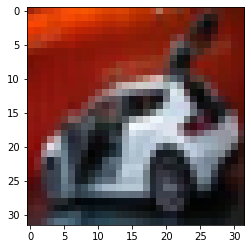}} &
{\includegraphics[width = 0.5in]{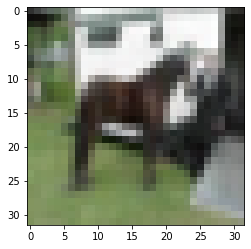}} &
{\includegraphics[width = 0.5in]{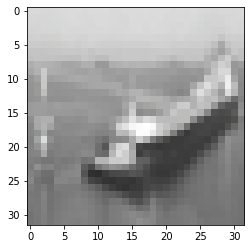}} &
{\includegraphics[width = 0.5in]{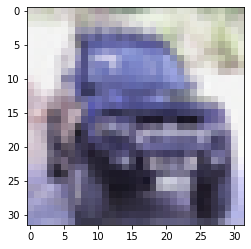}} &
{\includegraphics[width = 0.5in]{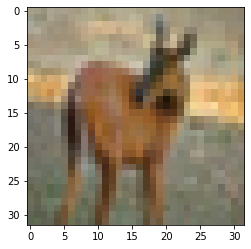}} \\
{\includegraphics[width = 0.5in]{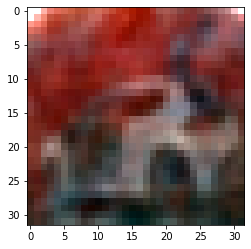}} &
{\includegraphics[width = 0.5in]{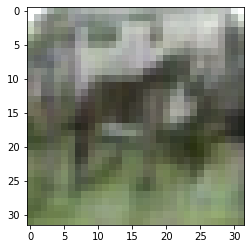}} &
{\includegraphics[width = 0.5in]{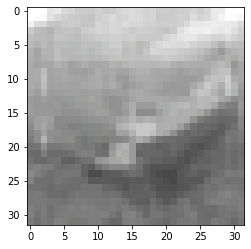}} &
{\includegraphics[width = 0.5in]{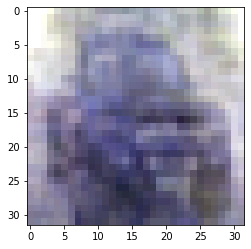}} &
{\includegraphics[width = 0.5in]{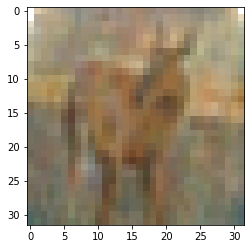}}\\
\end{tabular}
\caption{Reconstruction quality of CIFAR-10 images against various attacks. The first row shows the adversarial images and its reconstructed image is shown in second row.}
\label{Reconstruction quality of CIFAR-10 }
\end{figure}

\begin{table*}[ht!]
        
        \caption{Experimental results of various attack on MNIST and comparison with state of the art defenses}
        \begin{tabular}{p{1.5cm}p{1.5cm}p{1.5cm}p{1.5cm}p{1.5cm}p{1.5cm}p{1.5cm}p{1.5cm}p{1.5cm}}
                \hline
                Attack & Adversarial Acc.(\%) & Adversarial Loss & CAD Corrected Acc.(\%) & CAD Corrected Loss & CRD Acc.(\%)~\cite{dhaliwal2020compressive} & Madry et al. Acc.(\%)~\cite{madry2017towards} & Yu et al. Acc.(\%)~\cite{yu2018interpreting} & CMT (Black Box Setting) Acc.(\%)~\cite{wang2020defending}  \\
                \hline
                No Attack & - & - & 98.45 & 0.386 & 99.17 & 98.8 & 98.4 & 99.5 \\
                FGSM & 1.2 & 0.94 & 93.9 & 0.69 & - & 95.6 & 91.6 & 84.2\\
                PGD & 0.0 & 1.02 &  99.75 & 0.001  & - & 93.2 & - & -\\
                BIM & 0.0 & 1.054 &  99.7 & 0.014 & 74.7 & - & 88.1 & 79.5\\
                CW(L2) & 1.5 & 2.21 & 86.46 & 0.914 & 92.4 & 94 & 89.2 & 1.2 \\
                
                \hline
        \end{tabular}
        \label{table classification MNIST}
\end{table*}

\begin{table*}[ht!]
        
        \caption{Experimental results of various attack on CIFAR-10 and comparison with state of the art defenses}
        \begin{tabular}{p{1.5cm}p{1.5cm}p{1.5cm}p{1.5cm}p{1.5cm}p{1.5cm}p{1.5cm}p{1.5cm}p{1.5cm}}
                \hline
                Attack & Adversarial Acc.(\%) & Adversarial Loss & CAD Corrected Acc.(\%) & CAD Corrected Loss & CRD Acc.(\%)~\cite{dhaliwal2020compressive} & Madry et al. Acc.(\%)~\cite{madry2017towards}  & Yu et al. Acc.(\%)~\cite{yu2018interpreting} & CMT (Black Box Setting) Acc.(\%)~\cite{wang2020defending}   \\
                \hline
                No Attack & - & - & 84.41 & 0.916 & 84.9  & 87.3 & 83.1 & 80.1  \\
                FGSM & 0.0 & 1.602 & 75.33 & 1.201 & - & 56.1 & 68.5 & 81.8 \\
                PGD & 0.0 & 1.485 &  76.0 & 1.021  & - & 45.8 & - & -\\
                BIM & 0.0 & 1.487 & 78.66  & 0.926 & 49.4 & - & 62.7 & 80.2\\
                CW(L2) & 0.45 & 1.635 & 75.53 & 1.007 & 72.3 & 46.8 & 60.5 & 4.5 \\
                
                \hline
        \end{tabular}
        \label{table classification cifar10}
\end{table*}

\begin{table*}[ht!]
        
        \caption{Experimental result against $l_0$ norm attack JSMA on MNIST and CIFAR-10 }
        \begin{tabular}{p{1.5cm}p{1.5cm}p{1.5cm}p{1.5cm}p{1.5cm}p{1.5cm}p{1.5cm}p{1.5cm}p{1.5cm}}
                \hline
                Dataset & Perturbation size (t) & Adversarial Acc.(\%) & Adversarial Loss & CAD Corrected Acc.(\%) & CAD Corrected Loss & CRD Acc.(\%)~\cite{dhaliwal2020compressive} & Bafna et al. Acc.(\%)~\cite{bafna2018thwarting} & CMT (Black Box Setting) Acc.(\%)~\cite{wang2020defending}   \\
                \hline
                
                MNIST & 20 & 0.0 & 2.049 & 94.9 & 0.1483  & 55.9 & 90.8 & 81.3 \\
                CIFAR-10 & 35 & 12.4 & 10.344 &  71.6 & 7.01 & 67.3 & - & 79.3\\

                \hline
        \end{tabular}
        \label{table l0 attack}
\end{table*}

\begin{figure}
\centering
\begin{tabular}{c}
\subfloat{\includegraphics[width = 3.5 in, height = 2 in]{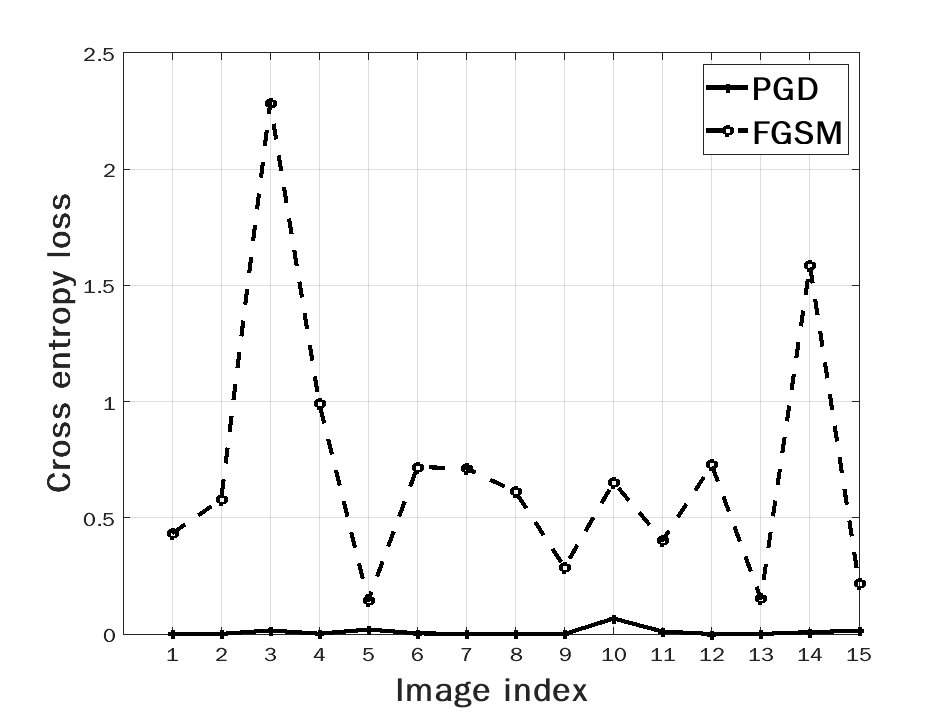}} \\
\subfloat{\includegraphics[width = 3.5 in, height = 2 in]{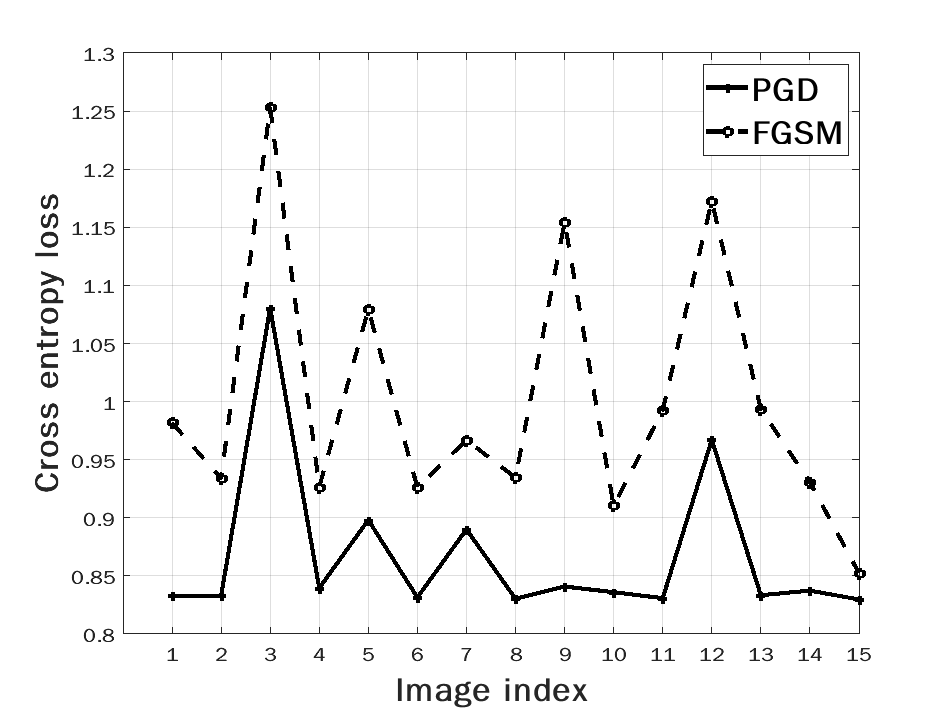}} \\
\end{tabular}
\caption{Cross entropy loss under FGSM and PGD attack for each 15 Image indexes. Top: MNIST, Bottom: CIFAR-10}
\label{PLots Cross entropy loss vs Image index}
\end{figure}

\subsection{Obfuscated gradients}
 Most of recently proposed defenses are suffer from the problem of obfuscated gradients \cite{athalye2018obfuscated},\cite{carlini2019evaluating}; the proposed defense often does not use accurate gradients while generating adversarial images for the testing phase. Here we argue that CAD algorithm does not cause gradient masking, the reasons being the following: 
 \begin{enumerate}
     \item Iterative attacks are usually superior to single step attacks. In order to verify this, we randomly select 15 images on which foolbox can craft a perturbed image. We choose FGSM and PGD as single step attack and iterative attack respectively,    evaluate each image separately on our model, and plot the  cross-entropy loss for each image. From Figure~\ref{PLots Cross entropy loss vs Image index} it can be seen clearly that, for each image, cross-entropy loss is always less for PGD attack compared to FGSM attack , for both MNIST and CIFAR-10. This matches the well-known fact that iterative attack is superior to single step attack. 
     \item We apply unbounded distortion for both FGSM and PGD and observe that each image is misclassified. Hence, the attack exhibits $100 \%$   success rate, which is another desired condition.
 \end{enumerate}

\section{Conclusion}\label{section:Conclusion}
In this paper, we have proposed a compressive sensing based adaptive defense (CAD) scheme. CAD algorithm chooses an appropriate recovery algorithm in each iteration using the multi-armed bandit theory, based on the observed nature of the residual error.  While  the standard basis pursuit algorithm was previously used to mitigate  $l_2$ attack, we have proposed a  modified basis pursuit with novel constraint to combat  $l_0$ and $l_\infty$ attacks, and also have provided their  performance bounds. The proposed CAD algorithm achieves excellent classification accuracy with low computational complexity and low  memory requirement for both white box gradient attacks and norm attacks. 

While our paper combines compressive sensing and multi-armed bandit techniques for adversarial image classification, this approach can be adopted even for classifying and detecting adversarial videos. However, computation complexity will be a major challenge for videos, and that can be alleviated to some extent by opportunistically sampling frames and applying tools similar to this paper on them. Tools from restless bandit theory can also be useful for videos. Thus, our paper opens the possibility of starting a new research domain on adversarial image and video detection using theoretical tools, which has traditionally seen mostly DNN and heuristic based efforts.
\bibliographystyle{unsrt}
\bibliography{bibliography.bib}

\end{document}